\documentclass[11pt,letterpaper]{article}

\usepackage[margin=1in]{geometry}
\usepackage{theorem,latexsym,graphicx,amssymb}
\usepackage{amsmath,enumerate}
\usepackage{subfigure}
\usepackage{float}
\usepackage{epsfig}
\usepackage{xspace}
\usepackage{paralist}
\usepackage{enumerate}
\usepackage{cases}
\usepackage{caption}
\usepackage{multicol}
\usepackage{graphicx}
\usepackage{xcolor}
\usepackage{tikz}
\usepackage{ifthen}
\usepackage{algorithm}
\usepackage[noend]{algpseudocode}

\newenvironment{proof}{{\bf Proof:  }}{\hfill\rule{2mm}{2mm}\vspace*{5pt}}
\newenvironment{proofof}[1]{{\vspace*{5pt} \noindent\bf Proof of #1:  }}{\hfill\rule{2mm}{2mm}\vspace*{5pt}}
\newtheorem{definition}{Definition}[section]
\newtheorem{corollary}{Corollary}
\newtheorem{theorem}{Theorem}[section]
\newtheorem{lemma}{Lemma}[section]

\newtheorem{fact}{Fact}[section]

\newcommand{\fomp}{\textsf{Fully Online Matching}\xspace}
\newcommand{\fobmp}{\textsf{Fully Online Bipartite Matching}\xspace}
\newcommand{\fofmp}{\textsf{Fully Online Fractional Matching}\xspace}
\newcommand{\obmp}{\textsf{Online Bipartite Matching}\xspace}

\newcommand{\ranking}{\textsf{Ranking}\xspace}

\newcommand{\eqdef}{\stackrel{\textrm{def}}{=}}
\newcommand{\expect}[2]{\operatorname{\mathbf E}_{#1}\left[#2\right]}
\newcommand{\vect}[1]{\ensuremath{\vec{#1}}}
\newcommand{\vecy}{\vect{y}}
\newcommand{\vecmv}[1][v]{\vect{y}_{\text{-}#1}}
\newcommand{\onev}{\ensuremath{\mathsf{1}}}

\newcommand{\opt}{\textsf{OPT}}

\newcommand{\pw}{p}

\newcommand{\wtf}{\textsf{Water-Filling}\xspace}

\title{Tight Competitive Ratios of Classic Matching Algorithms in the Fully Online Model }

\author{
	Zhiyi Huang\thanks{Department of Computer Science, The University of Hong Kong. Email: \texttt{\{zhiyi,zhtang,yhzhang2\}@cs.hku.hk}.} 
	\and Binghui Peng\footnote{Institute for Interdisciplinary Information Sciences, Tsinghua University. Part of the work was done while Binghui and Runzhou were visiting the University of Hong Kong. Email: \texttt{\{pbh15, trz15\}@mails.tsinghua.edu.cn}.}
	\and Zhihao Gavin Tang\footnotemark[1]
	\and Runzhou Tao\footnotemark[2]
	\and Xiaowei Wu\footnote{Department of Computer Science, City University of Hong Kong, Email: \texttt{wxw0711@gmail.com}.}
	\and Yuhao Zhang\footnotemark[1]}
\date{}

\begin{document}

\begin{titlepage}
	\thispagestyle{empty}
	\maketitle
	
	\begin{abstract}
	\thispagestyle{empty}

	Huang et al.~(STOC 2018) introduced the fully online matching problem, a generalization of the classic online bipartite matching problem in that it allows all vertices to arrive online and considers general graphs. They showed that the ranking algorithm by Karp et al.~(STOC 1990) is strictly better than $0.5$-competitive and the problem is strictly harder than the online bipartite matching problem in that no algorithms can be $(1-1/e)$-competitive.
	
	This paper pins down two tight competitive ratios of classic algorithms for the fully online matching problem. For the fractional version of the problem, we show that a natural instantiation of the water-filling algorithm is $2-\sqrt{2} \approx 0.585$-competitive, together with a matching hardness result.
	Interestingly, our hardness result applies to arbitrary algorithms in the edge-arrival models of the online matching problem, improving the state-of-art $\frac{1}{1+\ln 2} \approx 0.5906$ upper bound. 
	For integral algorithms, we show a tight competitive ratio of $\approx 0.567$ for the ranking algorithm on bipartite graphs, matching a hardness result by Huang et al. (STOC 2018). 
\end{abstract} 
\end{titlepage}

\section{Introduction}
\label{sec:intro}

Following the seminal work by Karp et al.~\cite{stoc/KarpVV90} that initiated the study of the \obmp problem by proposing the \ranking algorithm, online matching problems have drawn a lot of attentions in the online algorithm literature.
These problems have found numerous real-life applications, notably, in online advertising.
They are also the driving-force behind many important techniques for designing and analyzing online algorithms, including the randomized primal dual technique by Devanur et al.~\cite{soda/DevanurJK13}.

Recently, Huang et al.~\cite{stoc/HKTWZZ18} proposed a generalization of the \obmp problem called \fomp.
The generalization considers general graphs and allows all vertices to arrive online.
It captures a much wider family of real-life scenarios, including the ride-sharing problem.
Concretely, consider an undirected graph $G = (V, E)$.
Each step is either the arrival or the deadline of a vertex. 
At a vertex $v$'s arrival, all the edges between $v$ and those that arrive before $v$ are revealed. 
At its deadline, on the other hand, the algorithm must irrevocably either match it to an unmatched neighbor (if it is not matched already) or leave it unmatched. 
The model assumes that all neighbors of a vertex $v$ arrive before $v$'s deadline.
This turns out to be a natural condition when it comes to concrete scenarios such as ride-sharing.

Further, Huang et al.~\cite{stoc/HKTWZZ18} showed that the \fomp problem is quite intriguing from an algorithmic viewpoint in that 1) it takes a number of novel ideas to show that the \ranking algorithm by Karp et al.~\cite{stoc/KarpVV90} is strictly better than $0.5$-competitive even in the fully online setting, and 2) the fully online setting, even on bipartite graphs, is strictly harder than the original \obmp problem in that no algorithms can be $1 - 1/e \approx 0.632$-competitive.

\subsection{Our Contributions and Techniques.}

We develop better understandings on the \fomp problem by establishing two tight competitive ratios.
The first result considers the fractional version of the problem, where we are allowed to fractionally match each vertex to multiple neighbors so long as the total mass sum to at most one.
We show that the \wtf algorithm, which at each vertex's deadline matches its unmatched portion fractionally to all neighbors with smallest matched portion (i.e., the lowest water-level), gets a competitive ratio of $2 - \sqrt{2} \approx 0.585$.
We also construct a matching hard instance for \wtf. 
The hardness result applies to arbitrary algorithms if we consider edge arrival models~\cite{esa/BuchbinderST17}, even when preemptions are allowed~\cite{stacs/EpsteinLSW13, approx/McGregor05}, improving the best known bounds in these models.
The second result focuses on the integral problem and the \ranking algorithm.
We prove that its competitive ratio is exactly the $\Omega$ constant\footnote{This is the solution of $\Omega \cdot e^\Omega = 1$.} $\approx 0.567$ on bipartite graphs, improving the previous bound of $\approx 0.554$ and matching the previous hardness result by Huang et al.~\cite{stoc/HKTWZZ18}. 
See Figure~\ref{figure:result_comparison} for where our results sit compared with the previous works.

\begin{figure}
	\centering
	\includegraphics[width=\textwidth]{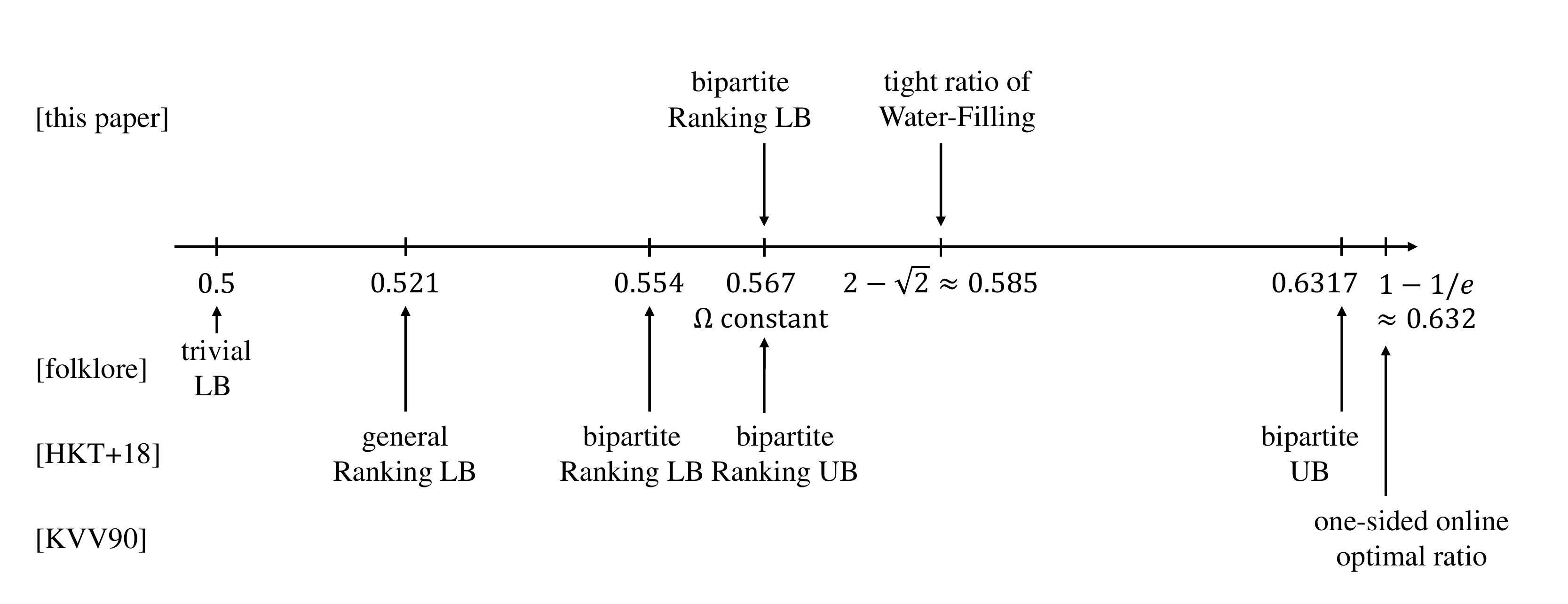}
	\caption{A comparison of the results in this paper and those in previous work. 
		LB means lower bound (algorithmic results), and UB means upper bound (hardness results).
	}
	\label{figure:result_comparison}
\end{figure}

\paragraph{Competitive Analysis of \wtf.}

The analysis of the \wtf algorithm is the relatively easy part of the paper. 
We follow the online primal dual framework by Buchbinder et al.~\cite{esa/BuchbinderJN07}, building on the notions of passive and active vertices by Huang et al.~\cite{stoc/HKTWZZ18}.

When a vertex $u$ matches another vertex $v$ at $u$'s deadline, Huang et al.~\cite{stoc/HKTWZZ18} referred to $u$ as the active vertex and $v$ as the passive vertex.
Intuitively, when edge $(u, v)$ is of concern, $v$ plays a role similar to an offline vertex in the \obmp problem since it sits back and allows $u$ to make the matching decision, while $u$ plays a role similar to an online vertex.
Following the same principle, for every vertex $v$, we refer to the portion that is matched before its deadline as the passive portion, and the portion that is matched at its deadline as the active portion. 

When a small portion $p$ of edge $(u, v)$ is chosen into the fractional matching, we shall split the gain of $p$ between the endpoints according to the current water-level $x_v$ of the passive vertex $v$ (i.e., the one with a later deadline).
For some function $g$ to be chosen in the analysis, $u$ shall get $\big( 1 - g(x_v) \big) \cdot p$ while $v$ shall get $g(x_v) \cdot p$.
Then, by an appropriate argument, we can lower bound the total gain of $u$ and $v$ by:
\begin{equation}
\label{eq:fractional_intro}
\int_0^{\pw_u} g(x) dx + \big(1-{\pw_u}\big)\big(1-g(x_v)\big) + \int_0^{x_v} g(x)dx.
\end{equation}
Here, $p_u$ is the passive portion of $u$, and $x_v$ is the passive portion of $v$ after $u$'s deadline.
The first term is the gain of $u$ due to its passive portion.
The second term lower bounds the gain of $u$ due to its active portion.
The third term lower bounds of the gain of $v$ due to its passive portion.

It remains to choose $g$ to maximize the above lower bound against the worst $p_u$ and $x_v$.
Unlike in the primal dual analysis of some other online matching problems, this is not exactly a standard ODE.
Nonetheless, we observe that it is almost symmetric w.r.t.\ $p_u$ and $x_v$. 
Indeed, choosing $g$ to be an appropriate linear function makes it symmetric and yields the optimal $2 - \sqrt{2}$ bound.

\paragraph{Matching Hardness for \wtf.}

Constructing a hard instance to show a matching $2 - \sqrt{2}$ upper bound on the competitive ratio of \wtf presents some technical obstacles beyond the existing techniques.
The construction is driven by Eqn.~\eqref{eq:fractional_intro}.
By our choice of $g$, Eqn.~\eqref{eq:fractional_intro} is equal to the lower bound $2 - \sqrt{2}$ only if $p_u$ and $x_v$ sum to precisely $2 - \sqrt{2}$.
Further, the performance of the algorithm is equal to the gain of the endpoints summing over all edges in the optimal matching in hindsight.
Therefore, a matching hard instance must satisfy that before the matching decision is made for an edge $(u, v)$ in the optimal matching, the water-levels of the two endpoints are prepared in advance so that the sum equals $2 - \sqrt{2}$.
This suggests that a tight instance for \wtf must look very different from the existing hard instances in the previous works (e.g., \cite{stoc/KarpVV90, stoc/DevanurJ12, stoc/HKTWZZ18}), where for every edge $(u, v)$ in the optimal matching, one of the two endpoints simply shows up with zero water-level and a matching decision is made for the edge.%
\footnote{It is easy to show that one cannot maintain at all time that some vertices have $2 - \sqrt{2}$ water-level while the other have $0$ with the \wtf algorithm.}

Our construction prepares the water-level of the vertices via a dynamic as follows.
It maintains at all time a set of vertices with some number of vertices at each water-level $x$ for $0 \le x \le 2 - \sqrt{2}$.
At each step, pick a vertex $u$ with an appropriate water-level $x_u$ and let it be $u$'s deadline.
Vertex $u$ connects to a subset of the vertices with water-level $2 - \sqrt{2} - x_u$, among which one vertex $v$ is $u$'s partner in the optimal matching.
After the step, $u$ and $v$ will be removed from the pool; 
new vertices (with zero water-level) will arrive to refill the pool if needed.
The matching decision of $u$ ``pumps up'' the water-level of all its neighbors to $2 - \sqrt{2} - x_u + \epsilon$.
Some of them will serve as the active endpoints with this water-level of some edges in the optimal matching;
some of them will serve as the passive counterparts;
the water-level of the remaining will be further ``pumped up'' by some vertex with water-level $x_u - \epsilon$.
We show how to maintain such a dynamic so that, in the long run, the endpoints of any edge in the optimal matching will have a total water-level close to $2 - \sqrt{2}$ when a matching decision is made for the edge.

\paragraph{Competitive Analysis of \ranking on Bipartite Graphs.}

We first explain why the previous analysis of Huang et al.~\cite{stoc/HKTWZZ18} is not tight on their hard instance.
Consider an edge $(u, v)$ in the optimal matching where $u$ has an earlier deadline. 
The previous analysis is tight only if there is a threshold $\theta$ such that whenever $v$'s rank is larger than $\theta$, $u$ is matched and $v$ is unmatched and, more importantly, whenever $v$'s rank is smaller than $\theta$, $v$ is passively matched and $u$ matches to the same vertex as in the previous case.
In the hard instance, however, $u$ is $v$'s only neighbor.
Therefore, if $u$'s own rank is sufficiently large such that $u$ matches actively, $u$ and $v$ will match each other when $v$'s rank is smaller than $\theta$.
Taking this extra gain into account gives the optimal ratio of $\approx 0.567$ for the hard instance.

Of course, we cannot na\"ively assume that one of the endpoints of any edge in the optimal matching will have only one neighbor.
The point is the previous approach that tries to characterize the matching status of $u$ and $v$ using a single threshold of $v$ cannot possibly capture the above extra gain.
We show that a good enough characterization in general takes three thresholds, a threshold of $u$ and two thresholds of $v$, one with $u$ in the graph and one without.
As a result, we get a new lower bound on the total expected gain of the endpoints that is strictly better than the previous one in all but a few bottleneck cases.
Then, we design a different gain sharing function that focuses on these bottleneck cases to obtain a tight analysis.

Finally, we remark that the three thresholds pin down when $u$ and $v$ match each other.
Previous works on online matching usually omit the gain from this case, which indeed happens with negligible probability in the worst case of those models (with a recent exception of \cite{icalp/HTWZ18}).
Our analysis shows it was just a lucky coincident that we do not need to consider the case when the endpoints match each other in those problems.
It becomes critical in a more general online matching model.

\subsection{Other Related Works}

Following Karp et al.~\cite{stoc/KarpVV90}, a series of works study different variants of the problem, including $b$-matching~\cite{tcs/KalyanasundaramP00}, adwords~\cite{jacm/MehtaSVV07,esa/BuchbinderJN07,stoc/DevanurJ12}, vertex-weighted matching~\cite{soda/AggarwalGKM11} and the random arrival model~\cite{stoc/KarandeMT11,stoc/MahdianY11,icalp/HTWZ18}.
Besides, the analysis of \ranking has been simplified in a series of papers~\cite{soda/GoelM08,sigact/BenjaminC08,soda/DevanurJK13}.

The \wtf algorithm has been studied to tackle several versions of the \obmp problems~\cite{esa/BuchbinderJN07, tcs/KalyanasundaramP00}. 
Devanur et al.~\cite{ec/Devanur0KMY13} considered the whole page optimization problem and extended the \wtf algorithm to use a carefully designed ``level function'' instead of a single water-level.
Wang and Wong~\cite{icalp/WangW15} considered an alternative model of \obmp that allows both sides of vertices to arrive online. 
They showed a $0.526$-competitive algorithm for a fractional version of the problem. 
Both analysis of~\cite{ec/Devanur0KMY13,icalp/WangW15} are based on the online primal dual framework by~\cite{esa/BuchbinderJN07}. 
This paper further illustrates the power of this framework for studying online fractional matching problems.

The hardness result in this paper improves the bounds for the following online matching models. 
In \textsf{online preemptive matching}~\cite{stacs/EpsteinLSW13, approx/McGregor05}, each edge arrives online and the algorithm must immediately decide whether to add the edge to the matching and to dispose of previously selected edges if needed. 
A harder edge-arrival model~\cite{esa/BuchbinderST17} forbids edge disposals. 
For both problems, the best previous bound stands at $\frac{1}{1+\ln 2}\approx 0.5906$~\cite{stacs/EpsteinLSW13}. 

Very recently, weighted variants of \fomp have been studied by~\cite{corr/abs-1803-01285,corr/abs-1806-10327}, both considering the ``windowed'' version of the problem, motivated by the ride-sharing applications.



\section{Preliminaries} \label{sec:preliminary}


%
We study both the fractional and the integral versions of \fomp. When the underlying graph is bipartite, we refer to the problem as \fobmp.
Consider the following standard linear program formulation of the matching problem and its dual.
\begin{align*}
\max: \quad & \textstyle \sum_{(u,v)\in E} x_{uv} && \qquad\qquad & \min: \quad & \textstyle\sum_{u \in V} \alpha_u\\
\text{s.t.} \quad & \textstyle \sum_{v:(u,v)\in E} x_{uv} \leq 1 && \forall u\in V & \text{s.t.} \quad & \alpha_u + \alpha_v \geq 1 && \forall (u,v)\in E \\
& x_{uv} \geq 0 && \forall (u,v)\in E & & \alpha_u \geq 0 && \forall u \in V
\end{align*}

\paragraph{Fractional Matching.}

In this setting, we may match edges fractionally. 
Let $x_{uv} \in [0, 1]$ be the fraction of edge $(u, v)$ in the matching. 
Assuming $u$ has an earlier deadline than $v$, this variable increases only at $u$'s deadline. 
We refer to it as \fofmp and study the classic \wtf algorithm (e.g., \cite{esa/BuchbinderJN07}) in this setting.
We give a formal definition of the algorithm below, in which the dual variables are updated as well. 
Note that the dual variables are used only in the analysis.
We fix an increasing function $g:[0,1]\to[0,1]$ to be specified later and use $x_u \eqdef \sum_{v:(u,v)\in E}x_{uv}$ to keep track of the water-level (i.e. total fractional mass) of $u$ at all time.
\begin{algorithm}
	\caption{The \wtf Algorithm}
	\label{alg:wtf}
	\begin{algorithmic}
		\State Initialize all $x_{uv}$'s and $\alpha_u$'s to be zero.
		\State {When the deadline of vertex $u$ is reached:}
		\State \quad Let $\pw_u=x_u$ be the water-level collected before $u$'s deadline. \Comment{$\pw_u$: passive water-level of $u$.}
		\State \quad Let $N(u)$ be the set of neighbors of $u$ whose deadlines are not reached.
		\State \quad \textbf{while} {$x_u < 1$ and $\min_{v \in N(u)}\{x_v\} < 1$} \textbf{do}	
		\State \quad\quad Allocate a $dx$ amount to each $x_{uv}$ for $v \in \arg\min_{v\in N(u)} \{x_v\}$.
		\State \quad\quad If $x_{uv}$ increases by $dx$, increase $\alpha_u$ and $\alpha_v$ respectively by
		\[
		d\alpha_u = (1-g(x_v))dx \quad \text{and} \quad d\alpha_v = g(x_v) dx.
		\]
	\end{algorithmic}
\end{algorithm}

We call the vertices in $N(u)$ the \emph{available} neighbors of $u$ at $u$'s deadline.
We further import the notions of active and passive vertices from~\cite{stoc/HKTWZZ18} and define them for both fractional and integral algorithms.
\begin{definition}[Active, Passive]
	For any edge $(u, v)$ that is (fractionally) matched by an algorithm at $u$'s deadline, we say that $u$ is active and $v$ is passive (w.r.t.\ edge $(u,v)$).
\end{definition}

\paragraph{Integral Matching.}

In this setting, $x_{uv}$'s must have binary values.
We will analyze the \ranking algorithm in Section~\ref{sec:bipartite} when the underlying graph $G$ is bipartite. 
Recall the definition of \ranking and some important notions from \cite{stoc/HKTWZZ18}.

\begin{algorithm}
	\caption{The \ranking Algorithm~\cite{stoc/HKTWZZ18}}
	\label{alg:ranking}
	\begin{algorithmic}
		\State (1) {a vertex $v$ arrives:}
		\State \phantom{(1)} pick $y_v \in [0,1)$ uniformly at random.
		\State (2) {a vertex $v$'s deadline is reached:}
		\State \phantom{(2)} \textbf{if} $v$ is unmatched,
		\State \phantom{(2) \textbf{if}} let $N(v)$ be the set of unmatched neighbors of $v$.
		\State \phantom{(2) \textbf{if}} \textbf{if} $N(v)=\emptyset$, \textbf{then} $v$ remains unmatched;
		\State \phantom{(2) \textbf{if}} \textbf{else} match $v$ to $\arg\min_{u\in N(v)} \{y_u\}$.
	\end{algorithmic}
\end{algorithm}

Let $M(\vecy)$ denote the matching produced when \ranking is run with $\vecy$ as the ranks.

\begin{definition}[Marginal Rank~\cite{stoc/HKTWZZ18}]
	For any $u$ and any ranks $\vecmv[u]$ of other vertices, the {marginal rank} $\theta$ of $u$ w.r.t. $\vecmv[u]$ is the largest value such that $u$ is passive in $M(y_u=\theta^\text{-}, \vecmv[u])$.
\end{definition}

The following is a restatement of Lemma 2.5 from \cite{stoc/HKTWZZ18} when restricted to bipartite graphs.
\begin{lemma}\label{lemma:alternating_path}
	In a bipartite graph, if $u$ is matched in $\vecy$, then from $M(\vecy)$ to $M(\vecmv[u])$, all neighbors of $u$ do not get better.
	Here, passive is better than active, which is in turns better than unmatched. 
	Conditioned on being passive, matching to a vertex with earlier deadline is better.
	Conditioned on being active, matching to a vertex with smaller rank is better.
\end{lemma}

We set primal variables according to \ranking. The randomized primal dual technique~\cite{soda/DevanurJK13} allows us to prove competitive ratio bounds through the following.
\begin{lemma}[\cite{stoc/HKTWZZ18}, Lemma 2.6]
	\label{lemma:dual_fitting}
	\ranking is $F$-competitive if we can set (non-negative) dual variables such that 1) $\sum_{(u,v)\in E} x_{uv} = \sum_{u \in V} \alpha_u$; and 2) $\expect{\vecy}{\alpha_u+\alpha_v} \geq F$ for all $(u,v)\in E$.
	%
	%
		%
		%
	%
\end{lemma}

\section{Tight Competitive Ratio of \wtf}

In this section, we give a tight analysis on the competitive ratio of the \wtf algorithm for the \fofmp problem.

\subsection{Lower Bound on the Competitive Ratio} \label{sec:wtf_lower}

We first prove that the competitive ratio of \wtf is at least $2-\sqrt{2}$.
Our approach is based on a primal dual analysis.

\begin{theorem}
	\label{thm:wtf_lower}
	\wtf is $(2-\sqrt{2})$-competitive.
\end{theorem}
\begin{proof}
	Recall that we update the primal variables according to \wtf and dual variables in a way that the dual objective always equals the primal objective. Using the standard primal dual technique, in order to prove that \wtf is $(2-\sqrt{2})$-competitive, it suffices to show that $\alpha_u + \alpha_v \ge 2-\sqrt{2}$ for all pairs of neighbors $u$ and $v$.
	
	Let $g(x) = \frac{\sqrt{2}}{2}x+1-\frac{\sqrt{2}}{2}$ be the function we used for defining dual variables.
	
	Fix any pair of neighbors $u,v$ where $u$ has an earlier deadline than $v$.
	Consider the moment right \textbf{after} $u$'s deadline. 
	It must be that either $x_u=1$ or $x_v=1$ (otherwise $x_u$ will further increase).
	As $v$ can only be matched passively, if $x_v=1$, we have
	\begin{equation*}
		\alpha_u + \alpha_v \ge \int_0^1 g(x) dx = 1 - \frac{\sqrt{2}}{4} \ge 2-\sqrt{2}.
	\end{equation*}
	
	Now suppose $x_u = 1$ and $x_v < 1$.
	Then, we have $\alpha_v = \int_0^{x_v} g(x)dx$.
	Next, consider the value of $\alpha_u$.
	Before $u$'s deadline, we have $\alpha_u = \int_0^{\pw_u} g(x)dx$ (recall that $\pw_u$ is the passive water-level of $u$).
	Since $x_u = 1$, and the water-level of $v$ after the deadline of $u$ is $x_v<1$,
	at any moment when the water-level of $u$ is increased from $\pw_u$ to $1$,
	the neighbor that $u$ matches has a water-level at most $x_v$.
	Hence, we have
	\begin{equation*}
	\alpha_u \ge \int_0^{\pw_u} g(x) dx + \int_{\pw_u}^1 (1-g(x_v))dx = \int_0^{\pw_u} g(x) dx + (1-\pw_u)(1-g(x_v)).
	\end{equation*}
	
	Summing the lower bounds on the two dual variables and by the definition of $g$, we have
	\begin{align*}
	\alpha_u+\alpha_v \geq & \int_0^{\pw_u} g(x) dx + (1-{\pw_u})(1-g(x_v)) + \int_0^{x_v} g(x)dx \\
	= & \frac{\sqrt{2}}{4}(\pw_u^2+x_v^2) + (1-\frac{\sqrt{2}}{2})(\pw_u+x_v) + (1-\pw_u)(\frac{\sqrt{2}}{2}-\frac{\sqrt{2}}{2}x_v) \\
	= & \frac{\sqrt{2}}{4}\left((\pw_u+x_v) - (2-\sqrt{2})\right)^2 + 2-\sqrt{2} \geq 2-\sqrt{2}.
	\end{align*}
	
	Hence, in both cases we have $\alpha_u+\alpha_v\geq 2-\sqrt{2}$, which gives the $2-\sqrt{2}$ lower bound on the competitive ratio of \wtf.
\end{proof}

\subsection{Upper Bound on the Competitive Ratio}
\label{sec:wtfhardness}

In this section we explicitly construct a hard instance, for which \wtf gives a solution of value $(2-\sqrt{2})\cdot \opt$.

\paragraph{Hard Instance.}
Let there be $2k\cdot m$ vertices, which are partitioned into $m$ groups of size $2k$.
For all $t\in[m]$, let the vertices in the $t$-th group be $U_t\cup V_t$, where $U_t = \{u_{t,1},\ldots,u_{t,k}\}$ and $V_t = \{v_{t,1},\ldots,v_{t,k}\}$.
Let $h:[0,1]\to [0,1]$ be a decreasing function\footnote{When $h(x)\equiv 1$, our instance becomes the $\frac{1}{1+\ln 2}\approx 0.5906$ hard instance by~\cite{stacs/EpsteinLSW13} for the edge arrival model.} (to be determined later) with $h(0)=1$ and $h(1)=0$. There are two types of edges in the graph (refer to Figure~\ref{figure:subgraph_induced}):
\begin{enumerate}
	\item[Upper triangle edges between $U_t$ and $V_t$:] $\forall t\in[m]$, $i\in[k]$ and $j\geq i$, $(u_{t,i}, v_{t,j}) \in E$;
	\item[$h$-induced edges between $U_t$ and $U_{t+1}$:] $\forall t\in[m-1]$, $i\in[k]$ and $j\leq \lfloor k\cdot h(\frac{i-1}{k}) \rfloor$, $(u_{t,i},u_{t+1,j}) \in E$.
\end{enumerate}

\begin{figure}[H]
	\centering
	\resizebox*{65mm}{!}
	{
	\begin{tikzpicture}
	[font=\bf,line width=0.3mm,scale=.8,auto=left,every node/.style={circle,draw}]
	\def \n {4}
	\foreach \i in {1,...,\n}
	{
		\node at (-4,1.4*\n-1.4*\i) (w\i) {$u_{2,\i}$};
	};
	\foreach \i in {1,...,\n}
	{
		\node at (1,1.4*\n-1.4*\i) (u\i) {$u_{1,\i}$};
	};
	\foreach \i in {1,...,\n}
	{
		\node at (6,1.4*\n-1.4*\i) (v\i) {$v_{1,\i}$};
		\foreach \j in {1,...,\i}
		{
			\draw (u\j) to (v\i);
		};
	};
	\draw (u1) to (w1);
	\draw (u1) to (w2);
	\draw (u1) to (w3);
	\draw (u1) to (w4);
	\draw (u2) to (w1);
	\draw (u2) to (w2);
	\draw (u2) to (w3);
	\draw (u3) to (w1);
	\draw (u3) to (w2);
	\draw (u3) to (w3);
	\draw (u4) to (w1);
	\end{tikzpicture}
	}
	\caption{Subgraph induced by $U_t\cup V_t \cup U_{t+1}$: illustrating example with $t=1$ and $k=4$}
	\label{figure:subgraph_induced}
\end{figure}
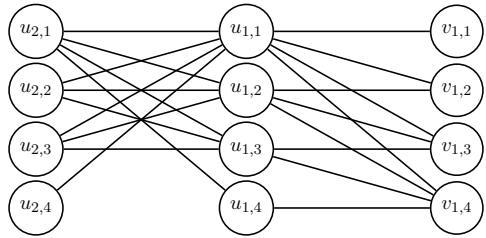

Finally, let the deadlines of the $u$ vertices be reached first, following the lexicographical order on $(t,i)$.
Then let the deadlines of the $v$ vertices be reached, i.e., after the deadline of $u_{m,k}$.%
\footnote{The relative order of the deadlines of $v$ vertices does not matter, as long as $v_{t,i}$'s deadline is after $u_{t,i}$'s deadline.}

It is easy to see that the hard instance is bipartite, where $(U_1\cup U_3\cup \dots) \cup (V_2\cup V_4\cup \dots)$ and $(U_2\cup U_4\cup \dots) \cup (V_1\cup V_3\cup \dots)$ are the two sides of vertices.
This graph admits a perfect matching, in which $u_{t,i}$ matches $v_{t,i}$ for all $t\in[m],i \in [k]$ and hence, $\opt = km$.


We first construct the function $h$ and prove the following technical lemma.
Let $c=2-\sqrt{2}$, and function $f:[0,c]\to [0,1]$ be defined as
\[
f(x) \eqdef \frac{1}{2}\left( \ln(1-x)+\ln(1-c+x) \right) + \frac{1}{\sqrt{2}(x-1)}+\frac{2+\sqrt{2}-\ln(1-c)}{2}.
\]
Let $\tau(x) \eqdef f^{-1}(x)$ and $h(x) \eqdef f(c-f^{-1}(x))$.
It is not difficult to see that $f$ is strictly decreasing. 
Hence, functions $h:[0,1]\to [0,1]$ and $\tau:[0,1]\to [0,c]$ are well defined.
Moreover, since $f(0) = 1$ and $f(c) = 0$, we have that $h$ is decreasing, $h(0) = 1$ and $h(1) = 0$, as required in the construction of the hard instance.
These functions might seem mysteries at this point, we will show a connection between the functions $h$ and $g$ via duality in Appendix~\ref{sec:wtf_gh_dual}, where $g$ is the gain sharing function that we used to define the dual variables in \wtf.


\begin{lemma}\label{lemma:property_of_h}
	For all $x\in [0,1]$ we have
	\begin{equation*}
	\int_0^x \frac{1-\tau(y)}{1-y+h(y)} dy = c-\tau(x),
	\quad \int_0^1 \tau(y)dy = 1-c
	\quad \text{and} \quad
	\int_0^1 \frac{1}{1-y+h(y)} dy < 1.
	\end{equation*}
\end{lemma}
\begin{proof}
	First we show the first equation, i.e., for all $x\in[0,1]$ we have $\int_0^x \frac{1-\tau(y)}{1-y+h(y)} dy = c-\tau(x)$.
	Note that $\tau(0) = c$ and, thus, both sides equal $0$ when $x = 0$.
	It suffices to check that for all $x\in[0,1]$, $\frac{1-\tau(x)}{1-x+h(x)} = -\tau'(x)$.
	Let $\phi = \tau(x)\in[0,c]$, we have $f(\phi) = x$ and $h(x) = f(c-f^{-1}(x)) = f(c-\phi)$.
	Then, we only need to check that
	\begin{equation*}
	\frac{1-\phi}{1-f(\phi)+f(c-\phi)} = -\tau'(x) = -\frac{1}{f'(\phi)},
	\end{equation*}
	which is true as $f$ is defined such that for all $\phi \in[0,c]$,
	\begin{equation*}
	1-f(\phi)+f(c-\phi)+(1-\phi)f'(\phi) = 0.
	\end{equation*}
	
	Taking integration from $0$ to $c$, the contributions of the 2nd and the 3rd terms cancel. 
	We have
	\begin{equation*}
		0 = c+\int_0^c (1-x)f'(x) dx = c - 1 + \int_0^c f(x) dx,
	\end{equation*}
	which implies the second equation because $\int_0^1 \tau(y)dy = \int_0^c f(x) dx = 1-c$, where the first equality follows because $\tau = f^{-1}$, $f$ is strictly decreasing, $f(0) = 0$, and $f(c) = 0$.
	
	Now we prove the last equation, i.e., $\int_0^1 \frac{1}{1-y+h(y)} dy < 1$.
	
	Observe that both $1-\tau(y)$ and $\frac{1}{1-y+h(y)}$ are increasing in terms of $y$.
	Hence we have
	\begin{equation*}
	c = c-\tau(1) = \int_0^1 \frac{1-\tau(y)}{1-y+h(y)} dy > \int_0^1 (1-\tau(y)) dy \cdot \int_0^1 \frac{1}{1-y+h(y)} dy = c\cdot \int_0^1 \frac{1}{1-y+h(y)} dy.
	\end{equation*}
	
	Dividing both sides by $c$ proves the last equation.
\end{proof}

Now we analyze the performance of \wtf on this instance. 
We first prove that by running \wtf on the hard instance, the passive water-levels of almost all vertices are strictly smaller than $1$.

\begin{lemma} \label{lemma:passive_<_1}
	For large enough $k$, \wtf produces a fractional matching with $p_{u_{t,i}} < 1$ for all $t \in [m],i\in[k]$ and $p_{v_{t,i}} < 1$ for all $t \in [m-1],i\in[k]$.
\end{lemma}
\begin{proof}
Observe that at the deadline of each $u_{t,i}$, where $t\in[m-1]$, it has $|N(u_{t,i}) \cap V_t| + |N(u_{t,i}) \cap U_{t+1}| = k-i+1+\lfloor k\cdot h(\frac{i-1}{k}) \rfloor$ neighbors whose deadlines are not reached. Moreover, as $h$ is decreasing, it is easy to see (by induction) that at the deadline of $u_{t,i}$, all available neighbors of $u_{t,i}$ have the same water-level.
Hence, \wtf increases the water-level of the available neighbors of $u_{t,i}$ at the same rate until $\min_{v\in N(u_{t,i})} x_v =1$ or $x_{u_{t,i}}=1$.

Since $u_{t+1,1}$ is a neighbor of every vertex in $U_t$, we have $p_{u_{t+1,1}} = \max_{j\in[k]} \{ p_{u_{t+1, j}},p_{v_{t, j}} \}$.
Therefore, it suffices to show that $p_{u_{t+1,1}}$ is smaller than $1$.
Note that each vertex $u_{t,i}$ has at most $1$ unit of unmatched portion that is distributed among $k-i+1+\lfloor k\cdot h(\frac{i-1}{k})\rfloor$ available neighbors and, thus, it increases the water-level of $u_{t+1,1}$ by at most $\frac{1}{k-i+1+\lfloor k\cdot h(\frac{i-1}{k})\rfloor}$.
Hence, when $k \to \infty$, we have
\begin{equation*}
p_{u_{t+1,1}} \le \sum_{i=1}^{k} \frac{1}{k-i+1+\lfloor k\cdot h(\frac{i-1}{k}) \rfloor} \to
\int_0^1 \frac{1}{1-y+h(y)} dy < 1,
\end{equation*}
where the last inequality follows from Lemma~\ref{lemma:property_of_h}. This finishes the proof.
\end{proof}

Lemma~\ref{lemma:passive_<_1} implies that, for large enough $k$, we can guarantee that when running \wtf on the hard instance, after the deadline of every $u_{t,i}$, where $t\in[m-1]$, we must have $x_{u_{t,i}} = 1$, as none of its neighbors with a later deadline has a water-level that reaches $1$.

\begin{corollary}\label{corollary:water-level-reaching-1}
	For all $t\in[m-1]$, we have $x_{u_{t,i}}=1$ after $u_{t,i}$'s deadline.
\end{corollary}

Now we are ready to prove the main theorem of this section.

\begin{theorem} \label{th:wthard}
	\wtf is at most $(2-\sqrt{2})$-competitive.
\end{theorem}
\begin{proof}
	Let $\mathbf{\pw}_t = (\pw_{u_{t,1}},\pw_{u_{t,2}},\ldots,\pw_{u_{t,k}})^\text{T}$ denote the passive water-level vector of $U_t$.
	Since the increment of matching at $u_{t,i}$'s deadline is at most $1-p_{u_{t,i}}$, the solution given by \wtf is
	\begin{equation*}
	\sum_{(u,v)\in E}x_{uv} \leq \sum_{t,i}(1-\pw_{u_{t,i}}) 
	= \sum_{t} (k-\|\mathbf{p}_t\|_1).
	\end{equation*}
	
	Indeed, by Corollary~\ref{corollary:water-level-reaching-1}, for all $t\in[m-1]$, the increment of matching at $u_{t,i}$'s deadline is exactly $1-p_{u_{t,i}}$.
	Recall that in the hard instance, $u_{t+1,i}$ is a neighbor of $u_{t,j}$ iff $\frac{i}{k} \leq h(\frac{j-1}{k})$.
	Hence we have
	\begin{equation*}
	\pw_{u_{t+1,i}} = \sum_{j=1}^{\lfloor k\cdot h^{-1}(\frac{i}{k})+1 \rfloor} \frac{1-p_{u_{t,j}}}{k-j+1+\lfloor k\cdot h(\frac{j-1}{k}) \rfloor}
	= \sum_{j=1}^{\lfloor k\cdot h^{-1}(\frac{i}{k})+1 \rfloor}(1-p_{u_{t,j}})\cdot a_j,
	\end{equation*}
	where $a_j = \frac{1}{k-j+1+\lfloor k\cdot h(\frac{j}{k}) \rfloor}$ is independent of $t$.
	In other words, there exists a $k\times k$ matrix $\mathsf{M}$ such that for all $t\in[m-1]$, $\mathbf{\pw}_{t+1} = \mathsf{M}(\mathbf{1}-\mathbf{\pw}_t)$.
	More precisely, we have $\mathsf{M}_{i,j} = a_j$ if $j\leq \lfloor k\cdot h^{-1}(\frac{i}{k})+1 \rfloor$, $\mathsf{M}_{i,j} = 0$ otherwise.
	Hence, for any $i\in[k]$, by Lemma~\ref{lemma:property_of_h}, we have
	\begin{equation*}
	\sum_{j\in [k]}\mathsf{M}_{i,j} \leq \sum_{j\in [k]}a_j < 1.
	\end{equation*}
	
	That is, $\mathsf{M}$ is a contraction matrix and the above mapping from $\mathbf{p}_{t}$ to $\mathbf{p}_{t+1}$ has a unique stationary vector $\mathbf{p^*}$, i.e. $\mathbf{p^*}= \mathsf{M}(\mathbf{1}-\mathbf{p^*})$. Moreover, $\lim_{t \to \infty}\mathbf{\pw}_t = \mathbf{p^*}$\footnote{Observe that $(\mathbf{\pw}_{t+1}-\mathbf{p^*})= \mathsf{M}(\mathbf{p^*}-\mathbf{\pw}_t)$ and $\mathsf{M}$ is a contraction matrix.}. Thus, for any fixed $k$, when $m\to \infty$, the ratio between the matching size of \wtf and the optimal is
	\[
	\lim_{m\to \infty} \frac{\sum_{t} (k-\|\mathbf{p}_t\|_1)}{m} = 1-\frac{1}{k}\cdot\|\mathbf{p^*}\|_1.
	\]
	
	Finally, we consider when $k\to \infty$ and calculate the stationary vector. 
	In this case, $\mathbf{p^*}$ becomes a function $p:[0,1]\to[0,1]$ and the linear equation $\mathbf{p^*}= \mathsf{M}(\mathbf{1}-\mathbf{p^*})$ becomes the following
	\[
	\int_0^{h^{-1}(x)} \frac{1-p(y)}{1-y+h(y)} dy = p(x), \quad \forall x \in [0,1].
	\]
	
	We verify that $p=\tau$ is a solution to this system of equations by Lemma~\ref{lemma:property_of_h}. For all $x$, we have
	\begin{align*}
		\int_0^{h^{-1}(x)} \frac{1-\tau(y)}{1-y+h(y)} dy	=  c-\tau(h^{-1}(x))
		= \tau\left( f\left(c-\tau(h^{-1}(x))\right) \right) = \tau\left(h\left(h^{-1}(x)\right)\right) = \tau(x).
	\end{align*}
	
	Thus, the ratio between \wtf and \opt \ is $1-\int_0^1 \tau(y) dy = c = 2 - \sqrt{2}$.
\end{proof}

Interestingly, we show that our hardness result applies to the edge-arrival models of the online matching problems.
In the \textsf{Online Edge Arrival Matching} problem~\cite{esa/BuchbinderST17}, at each step, an edge arrives online and the algorithm must irrevocably decide whether to add the edge to the matching; in the preemptive setting (\textsf{Online Preemptive Matching}~\cite{stacs/EpsteinLSW13, approx/McGregor05}), instead, we are allowed to dispose of edges in the matching before accepting a new edge.

\begin{corollary} \label{cor:hard_other}
	No algorithm can be better than $(2-\sqrt{2})$-competitive for \textsf{Online Edge Arrival Matching} and \textsf{Online Preemptive Matching}, even if fractional matching is allowed.
\end{corollary}
\begin{proof}
	Since the edge arrival model (resp. integral matching) is strictly harder than the preemptive model (resp. fractional matching), it suffices to consider the second model with fractional matching.
	Consider the previous hard instance with the following modifications.
	The underlying graph remains the same and each vertex is associated with the same deadline as before.
	At $u_{t,i}$'s deadline, its incident edges with available neighbors are revealed one by one.
	In this way, all available neighbors of $u_{t,i}$ are indistinguishable at this moment, i.e. they share the same set of neighbors.
	Thus by assigning random identities to these vertices, the available neighbors of $u_{t,i}$ have the same expected increment in matched fraction. 
	Moreover, since no edge incident to each vertex comes after its deadline, it is not beneficial for an algorithm to dispose of previously chosen edges. Therefore, no algorithm can do better than \wtf in expectation and the lower bound $2-\sqrt{2}$ applies.
\end{proof}

\section{Tight Competitive Ratio of \ranking on Bipartite Graphs} \label{sec:bipartite}

Let $\Omega\approx 0.5671$ denote the Omega constant, which is the solution for the equation $\Omega\cdot e^\Omega = 1$.
In this section, we prove that \ranking is $\Omega$-competitive for the \fobmp, matching the $\Omega$ hardness result given by Huang et al.~\cite{stoc/HKTWZZ18}.

\begin{theorem}
	\label{thm:ranking_bip}
\ranking is $\Omega$-competitive for \fobmp. 	
\end{theorem}

We adopt the randomized primal dual analysis from \cite{stoc/HKTWZZ18}.
Recall the dual assignment that distributes the gain of each matched edge between its two endpoints as follows.
\begin{itemize}
	\item \emph{Gain Sharing:} Whenever a pair $(u, v)$ is matched with $u$ being active and $v$ being passive, let $\alpha_u = 1 - g(y_v)$ and $\alpha_v = g(y_v)$, where $g : [0,1] \rightarrow [0,1]$ is non-decreasing, and $g(1) = 1$.
\end{itemize}

By Lemma~\ref{lemma:dual_fitting}, it suffices to prove that $\expect{\vecy}{\alpha_u+\alpha_v} \ge \Omega$ for all pairs of neighbors $u,v$.
Suppose $u$ has an earlier deadline than $v$ and $\vecmv$ is the rank vector of all vertices excluding $v$. Let $\theta$ be the marginal rank of $v$. The following lemma lies in the central of the proof by \cite{stoc/HKTWZZ18}.

\begin{fact}[\cite{stoc/HKTWZZ18}, Lemma 3.2] 	\label{fact:ranking}
	For any arbitrarily fixed $\vecmv$, we have
	\begin{equation*}
	\expect{y_v}{\alpha_u +\alpha_v} \ge \min_{\theta\in[0,1]} \left\{ \int_{0}^{\theta}g(y_v)dy_v + \min \left\{1-g(\theta), g(y_u) \right\} \right\}.
	\end{equation*}
\end{fact}

Our main technical contribution is an improved version of the above lower bound.
Indeed, using Fact~\ref{fact:ranking} as a lower bound, one cannot achieve a competitive ratio greater than $0.56$ by just optimizing $g$.%
\footnote{The function $g$ is not optimized in \cite{stoc/HKTWZZ18} with respect to their lower bound. However, the ratio is less than $0.56$ with the optimal $g$ function.}
In the following, we will first illustrate how this lower bound can be improved for the hard instance given in~\cite{stoc/HKTWZZ18}.
Then, we show in Section~\ref{subsec:rank_complete} how to prove the $\Omega$ competitive ratio for general instances.

\subsection{Better Competitive Ratio for the Hard Instance}
\label{sec:ranking_instance}
Recall the following hard instance for \ranking that is given by~\cite{stoc/HKTWZZ18}.
In the instance (refer to Figure~\ref{figure:ranking_hard_instance}), the vertices are organized into (infinitely many) groups of size $2k$, where each group $U_t\cup V_t$ induces a perfect matching.
For all $t\in[m-1]$, the vertices $U_t$ and $U_{t+1}$ are connected by a complete bipartite graph.
The deadline of every $u_{t,i}$ is earlier than $v_{t,i}$, and deadlines of $u_{t,i}$ follow the lexicographic order on $(t,i)$.

\begin{figure}[htp]
	\centering
	\resizebox*{45mm}{!}
	{
	\begin{tikzpicture}
	[font=\bf,line width=0.3mm,scale=.8,auto=left,every node/.style={circle,draw}]
	\def \n {4}
	\def \X {3}
	\node at (1,1.3*\n+0.5) (prev) {Prev};
	\node at (6,1.3*\n+0.5) (prevv) {Prev};
	\foreach \i in {1,...,\n}
	{
		\node at (1,1.3*\n-1.3*\i) (u\i) {$u_{t,\i}$};
	};
	\foreach \i in {1,...,\n}
	{
		\node at (6,1.3*\n-1.3*\i) (v\i) {$v_{t,\i}$};
		\draw (u\i) to (v\i);
	};
	\node at (1,-1.3-0.5) (next) {Next};
	\node at (6,-1.3-0.5) (nextv) {Next};
	\draw (prevv) to (prev);
	\draw (next) to (nextv);
	\foreach \i in {1,...,\n}
	{
		\draw[out=135,in=225] (u\i) to (prev);
		\draw[out=135,in=225] (next) to (u\i);
	};
	\end{tikzpicture}
	}
	\caption{Hard instance of \ranking: illustrating example with $k=4$.}
	\label{figure:ranking_hard_instance}
\end{figure}
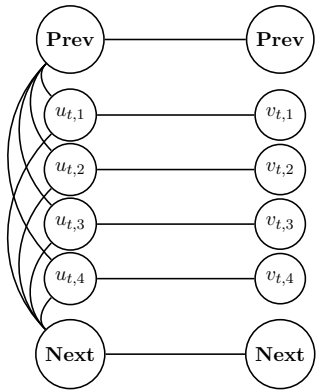

It is shown in~\cite{stoc/HKTWZZ18} that when running \ranking on the above instance,
at the deadline of the first vertex of each group, e.g., $u_{t,1}$, the expected fraction of unmatched vertices in $U_t$ (which is also the competitive ratio of \ranking) is given by the equation $x = e^{-x}$. In other words, the competitive ratio of \ranking is $\Omega$ on the above instance (when $k\rightarrow \infty$).

In the following, we show that the competitive ratio of \ranking is $\Omega$, using the randomized primal dual framework, and explain what is missing in the previous analysis.
Fix any pair of neighbors $u,v$ in the same group s.t.\ $u$ has an earlier deadline than $v$.
Next, we fix the ranks of all vertices but $v$ arbitrarily, and lower bound $\expect{y_v}{\alpha_u +\alpha_v}$ for any edge $(u, v)$ that appears in the perfect matching\footnote{Note that the competitive ratio equals $\sum_{u} \expect{}{\alpha_u} = \sum_{(u,v) \text{ appears in the perfect matching}} \expect{}{\alpha_u + \alpha_v}$.}.

Observe that $u$ is the only neighbor of $v$.
If $u$ is passive, then $v$ is unmatched regardless of $y_v$, which implies $\expect{y_v}{\alpha_u+\alpha_v}=g(y_u)$.
Otherwise, let $\theta$ be the marginal rank of $v$. By definition, when $y_v > \theta$, $u$ matches a vertex with rank $\theta$ and hence $\alpha_u=1-g(\theta)$. 
For the case when $y_v < \theta$, it is shown in~\cite{stoc/HKTWZZ18} (using Lemma~\ref{lemma:alternating_path}) that $u$ does not get worse: $u$ either is passive, or actively matches a vertex with rank at most $\theta$. That is, $\alpha_u \ge \min \{g(y_u), 1-g(\theta)\}$ when $y_v < \theta$.
However, for the specific hard instance given in Figure~\ref{figure:ranking_hard_instance}, $u$ is $v$'s only neighbor.
Hence, $u$ and $v$ will match each other when $y_v < \theta$.
Therefore, we have 
\[
\expect{y_v}{\alpha_u + \alpha_v} = \int_0^\theta (\alpha_u + \alpha_v) d y_v  + \int_\theta^1 \alpha_u d y_v = \theta + (1-\theta) \cdot (1-g(\theta)).
\]
Together with the case when $u$ is passive, we have that
\[
\expect{y_v}{\alpha_u+\alpha_v} \ge f(y_u) \eqdef \min \left\{ g(y_u), \min_{\theta \in [0,1]}\{\theta + (1-\theta) \cdot (1-g(\theta))\} \right\}.
\]
This bound is strictly stronger than Fact~\ref{fact:ranking}, as we fully characterize the gain of $\alpha_u$ when $y_v$ is smaller than its marginal rank, rather than the loose lower bound $\min \{g(y_u), 1-g(\theta)\}$ given in~\cite{stoc/HKTWZZ18}. By taking expectation over $y_u$ and optimizing the function $g(\cdot)$ (see Section~\ref{subsec:rank_complete}), the above lower bound implies that \ranking is $\Omega$-competitive on the hard instance.

In general, $v$ does not necessarily match $u$ when $y_v < \theta$.
However, when this fails to happen, we are able to retrieve extra gain of $\alpha_v$ when 
$u$ is passive. (Recall that in the hard instance, $v$ is unmatched when $u$ is passive.) The complete analysis involves a more careful treatment that considers the randomness of $y_u,y_v$ at the same time, when deriving the lower bound.

\subsection{Proof of Theorem~\ref{thm:ranking_bip}} \label{subsec:rank_complete}

Consider any neighboring vertices $u$ and $v$.
In the following, we fix an arbitrary assignment of ranks to all vertices but $u,v$. We denote this assignment of ranks by $\vecmv[uv]$.
Unless otherwise specified, we use $\expect{}{\cdot}$ to denote the expectation taken over the randomness of $y_u$ and $y_v$.

Instead of using a single threshold $\theta$ of $v$ as in the previous analysis, we will make use of multiple thresholds to give a good enough characterization of the matching status of $u$ and $v$ in order to derive the tight competitive ratio.
We introduce the first two below.

\begin{definition}[$\tau$ and $\gamma$]
	Consider the graph $G-\{v\}$ with $v$ removed.
	Let $\tau$ be the marginal rank of $u$ w.r.t. $\vecmv[uv]$.
	In other words, $u$ is passive iff $M(y_u < \tau, \vecmv[uv])$.
	Similarly, let $\gamma$ be the marginal rank of $v$ w.r.t. $\vecmv[uv]$ in graph $G-\{u\}$, i.e., with $u$ removed.
\end{definition}

\begin{lemma}\label{lemma:above}
	$\expect{}{\alpha_u \cdot\onev(y_u < \tau) + \alpha_v\cdot \onev(y_v < \gamma)} = \int_0^\tau g(y_u)dy_u 
	+ \int_0^\gamma g(y_v)dy_v$
\end{lemma}
\begin{proof}
	Consider $y_u = y < \tau$. 
	By the definition of $\tau$, we know that for all $y_v\in[0,1]$, $u$ is passive in $M(y_u=y,y_v, \vecmv[uv])$, because inserting $v$ (with any rank) to the graph cannot make $u$ worse (by Lemma~\ref{lemma:alternating_path}).
	Thus, for all $y_u < \tau$ and $y_v\in[0,1]$, we have $\alpha_u = g(y_u)$, which correspond to the first term of the RHS.
	For the same reason, for all $y_u\in[0,1]$, $v$ is passive in $M(y_v<\gamma,y_u, \vecmv[uv])$, which gives $\alpha_v = g(y_v)$, and the second term of the RHS.
\end{proof}

For all $y_u\in[0,1]$, let $\theta(y_u)$ be the marginal rank of $v$ w.r.t.\ $\vecmv[v] = (y_u,\vecmv[uv])$.
Recall $v$ is always passive (regardless of $y_u$) when $y_v < \gamma$.
Hence, we have $\theta(y_u) \geq \gamma$ for all $y_u\in[0,1]$.


\begin{lemma}\label{lemma:below}
	For any fixed $y_u > \tau$, we have
	\begin{equation*}
	\expect{y_v}{\alpha_u + \alpha_v\cdot \onev(y_v > \gamma)} \geq 1-\gamma - (1-\theta(y_u))\cdot g(\theta(y_u)) 
	+ \gamma\cdot \min\{g(y_u), 1-g(\theta(y_u))\}.
	\end{equation*}
\end{lemma}
\begin{proof}
	By the definition of $\theta(y_u)$, we know that when $y_v = \theta(y_u)^+$ (slightly larger than $\theta(y_u)$), $v$ is not passive.
	Thus, $u$ must be matched.
	Moreover, $u$ must be active.
	Otherwise $u$ should remain passive when $v$ is removed, because the deadline of $v$ is later than $u$, which contradicts the definition of $\tau$ (recall that we fix some $y_u > \tau$).
	Hence, when $y_v = \theta(y_u)^+$, $u$ actively matches some vertex with rank at most $\theta(y_u)$.
	As increasing the rank of $v$ does not create any difference to the final matching, for all $y_v > \theta(y_u)$, we have $\alpha_u \geq 1-g(\theta(y_u))$.
	
	Note that it is possible that $\theta(y_u) = 1$, i.e., $v$ is passive for all rank $y_v \in [0,1]$, in which case the above lower bound still holds.
	Since the graph is bipartite, by Lemma~\ref{lemma:alternating_path}, for all $y_v< \theta(y_u)$, we have $\alpha_u \geq \min\{g(y_u), 1-g(\theta(y_u))\}$.
	
	Finally, we show that for any $y_v\in (\gamma, \theta(y_u))$, we have $\alpha_u + \alpha_v \geq 1$.
	Fix any $y_v \in (\gamma,\theta(y_u))$.
	By definition $v$ is passive. Consider the first moment when one of $u,v$ is matched.
	
	Suppose at this moment, $v$ is matched (passively) by some vertex $z$.
	Then, we show that $z = u$, which gives $\alpha_u + \alpha_v = 1$.
	Otherwise, $z$ must have an earlier deadline than $u$.
	Then, we know that $v$ remains passive with $u$ removed, which contradicts the definition of $\gamma$.
	
	Suppose at this moment, $u$ is matched.
	Then we know that $u$ must active, as otherwise $u$ remains passive with $v$ removed, which contradicts the definition of $\tau$.
	Suppose $u$ matches some vertex $z$.
	Since $v$ is not matched at this moment, the rank of $z$ is no more than $y_v$, which implies $\alpha_u \geq 1-g(y_v) = 1-\alpha_v$, as required.

	To sum up, for any fixed $y_u > \tau$, we have
	\begin{eqnarray*}
	\expect{y_v}{\alpha_u + \alpha_v\cdot \onev(y_v > \gamma)} 
	& \geq & 
	\int_0^\gamma \alpha_u d y_v + \int_\gamma^{\theta(y_u)} (\alpha_u + \alpha_v) d y_v + \int_{\theta(y_u)}^1 \alpha_u d y_v \\
	& \geq & 
	\vphantom{\bigg(}
	\gamma\cdot \min\{g(y_u), 1-g(\theta(y_u))\} + (\theta(y_u) - \gamma) + (1-\theta(y_u))\cdot (1-g(\theta(y_u))) \\
	& \geq & 
	\vphantom{\bigg(}
	1-\gamma - (1-\theta(y_u))\cdot g(\theta(y_u)) + \gamma\cdot \min\{g(y_u), 1-g(\theta(y_u))\},
	\end{eqnarray*}
	as claimed.
\end{proof}

Combing the two lemmas, we have the following lower bound. Observe that the following bound degrades to the one we derived for the hard instance in Subsection~\ref{sec:ranking_instance}, when $\gamma=0$.

\begin{lemma}\label{lemma:bipartite_main}
	For any neighbor $u$ of $v$ that has an earlier deadline than $v$, and for any $\vecmv[uv]$, we have
	\begin{align*}
	\expect{}{\alpha_u + \alpha_v} \geq  \min_{\tau, 0\leq\gamma\leq\theta\leq 1} \Bigg\{ \int_0^\tau g(y_u)dy_u 
	&+ \int_0^\gamma g(y_v)dy_v + (1-\tau)(1-\gamma - (1-\theta)\cdot g(\theta))\\
	& + \gamma\cdot\int_\tau^1 \min\{1-g(\theta),g(y_u)\} dy_u \Bigg\}
	\end{align*}
\end{lemma}
\begin{proof}
	First, we show that there exists $\theta$ such that $\theta(y_u) = \theta$ for all $y_u > \tau$.
	Consider the graph with $v$ removed, and let $y_u = \tau^+$.
	By the definition of $\tau$, $u$ is not passive.
	\begin{enumerate}
		\item If $u$ is unmatched, then we know that after inserting $v$ with any $y_v \in [0,1]$, $v$ is passive, as otherwise $u$ will be matched with $v$ removed.
		Hence, we have $\theta(y_u) = 1$ for all $y_u> \tau$;
		\item Otherwise, $u$ is active.
		Let $\theta = \theta(\tau^+)$.
		Then, we know that $v$ is not passive when inserted to the graph with $y_v = \theta^+$.
		Moreover, we know that $u$ is active after the insertion: if $u$ is passive, then $u$ remains passive with $v$ removed, which contradicts the definition of $\tau$.
		Since increasing $y_u$ does not change the matching, we have $\theta(y_u) \leq \theta$ for all $y_u > \tau$.
		On the other hand, when $y_v = \theta^\text{-}$ and $y_u = \tau^+$, $u$ is active and $v$ is passive.
		Since increasing $y_u$ does not change the matching, we have $\theta(y_u) \geq \theta$ for all $y_u > \tau$.
		The sandwiching bounds imply that $\theta(y_u) = \theta$ for all $y_u > \tau$.
	\end{enumerate}

	Hence, combining Lemma~\ref{lemma:above} and~\ref{lemma:below}, we have
	\begin{align*}
	& \expect{}{\alpha_u+\alpha_v} \geq \expect{}{\alpha_u\cdot \onev(y_u<\tau) + \alpha_v\cdot \onev(y_v<\gamma)}
	+\int_\tau^1 \expect{y_v}{\alpha_u+\alpha_v\cdot \onev(y_v > \gamma)} d y_u \\
	= & \int_0^\tau g(y_u) dy_u + \int_0^\gamma g(y_v) d y_v + \int_{\tau}^1 \Big( 1-\gamma - (1-\theta)\cdot g(\theta) + \gamma\cdot \min\{g(y_u), 1-g(\theta)\}  \Big) d y_u \\
	= & \int_0^\tau g(y_u) dy_u + \int_0^\gamma g(y_v) d y_v + (1-\tau)\cdot\big(1-\gamma - (1-\theta)\cdot g(\theta)\big) + \gamma\cdot\int_{\tau}^1  \min\{g(y_u), 1-g(\theta)\} d y_u.
	\end{align*}
	
	Taking minimum over $\tau$ and $\gamma \leq \theta$ gives Lemma~\ref{lemma:bipartite_main}.
\end{proof}

\begin{proofof}{Theorem~\ref{thm:ranking_bip}}
	Fix the non-decreasing function $g$ as follows:
	\begin{equation*}
	g(y) = \begin{cases}
	\frac{c}{1-y}, \qquad&\text{when } y< \frac{1-2c}{1-c},\\
	1-c, \qquad&\text{when } \frac{1-2c}{1-c}\leq y <1,\\
	1, \qquad&\text{when } y = 1,
	\end{cases}
	\end{equation*}
	where $c = \frac{1}{1+e^\Omega}\approx 0.3619$.
	Let $f(\tau, \gamma, \theta)$ denote the expression to be minimized on the RHS of Lemma~\ref{lemma:bipartite_main}.
	Then, we have
	\begin{align*}
		f(\tau, \gamma, \theta)  = \int_0^\tau g(y_u) dy_u & + \int_0^\gamma g(y_v) dy_v + (1-\tau)\Big(1-\gamma-(1-\theta)\cdot g(\theta)\Big)\\
		& + \gamma\cdot \int_\tau^1 \min\{g(y_u),1-g(\theta)\} dy_u.
	\end{align*}
	
	Fix any $\gamma$ and $\theta$, and suppose $g(\tau) < 1-g(\theta)$, then observe that
	\begin{eqnarray*}
	\frac{\partial f(\tau,\gamma,\theta)}{\partial \tau} 
	& = & 
	g(\tau) - (1-\gamma-(1-\theta)\cdot g(\theta)) -\gamma\cdot \min\{g(\tau),1-g(\theta)\} \\
	& = & (1-\gamma)\cdot g(\tau) - (1-\gamma) + (1-\theta)\cdot g(\theta) \\
	& < & (1-\gamma) (1 - g(\theta)) - (1-\gamma) + (1-\theta)\cdot g(\theta) \\
	& = & (\gamma - \theta)\cdot g(\theta) \leq 0.
	\end{eqnarray*}
	Here, the last inequality holds because we have $\theta \geq \gamma$ by their definitions.
	
	Thus, the minimum of $f(\tau,\gamma,\theta)$ over $\tau\in[0,1]$, $0\leq \gamma\leq \theta\leq 1$ must be obtained when $g(\tau) \geq 1-g(\theta)$.
	As a result, we get that
	\begin{equation*}
	f(\tau,\gamma,\theta) = \int_0^\tau g(y_u) dy_u + \int_0^\gamma g(y_v) dy_v + (1-\tau)\Big( 1-(1-\theta+\gamma)\cdot g(\theta) \Big).
	\end{equation*}
	
	If we relax the constraint that $\theta \geq \gamma$, then the maximum of $(1-\theta+\gamma)g(\theta)$ is achieved when $\theta^* = \frac{1-2c}{1-c}$ (for which $g(\theta^*) = 1-c$). 
	Note that the maximum is $(\frac{c}{1-c}+\gamma)\cdot(1-c) = (1-c)\cdot\gamma + c$, which is greater than the value of expression when $\theta = 1$, i.e., $\gamma$.
	Thus, we have
	\begin{equation*}
	f(\tau,\gamma,\theta) \geq f(\tau,\gamma,\theta^*) = \int_0^\tau g(y_u) dy_u + \int_0^\gamma g(y_v) dy_v + (1-\tau)\cdot(1-\gamma)\cdot (1-c).	
	\end{equation*}
	
	It is easy to see that the minimum of $f(\tau,\gamma,\theta^*)$ must be achieved when $\gamma < \theta^* = \frac{1-2c}{1-c}$ (for which $g(\gamma) < 1-c$), as otherwise the partial derivative
	\begin{equation*}
	\frac{\partial f(\tau,\gamma,\theta^*)}{\partial \gamma} = g(\gamma)-(1-\tau)\cdot(1-c) \geq 0.
	\end{equation*}
	
	Since $f(\tau,\gamma,\theta^*)$ is symmetric for $\tau$ and $\gamma$, the same conclusion holds for $\tau$, which means
	\begin{align*}
	f(\tau,\gamma,\theta^*) = & \int_0^\tau \frac{c}{1-x} dx + \int_0^\gamma \frac{c}{1-x} dx + (1-\tau)\cdot(1-\gamma)\cdot (1-c) \\
	= & -c \ln(1-\tau) -c \ln(1-\gamma) + (1-\tau)\cdot(1-\gamma)\cdot(1-c) \\
	\geq & c-c\cdot\ln(\frac{c}{1-c}) = \frac{1+\Omega}{1+e^\Omega} = \Omega,
	\end{align*}
	where the inequality comes from the fact that (take $(1-\tau)\cdot(1-\gamma)$ as the variable) function $(1-c)\cdot x-c\cdot\ln(x)$ achieves its minimum when $x = \frac{c}{1-c}$.
\end{proofof}

{
	\bibliography{matching}
	\bibliographystyle{alpha}
}

\newpage
\appendix

\section{Primal-Dual Connection between the Upper and Lower Bounds} \label{sec:wtf_gh_dual}

We provide an interesting primal-dual connection between the primal dual analysis in Section~\ref{sec:wtf_lower} and the hard instance in Section~\ref{sec:wtfhardness}, which inspires us to find the correct $h$ function in Section~\ref{sec:wtfhardness}.

Recall the following lower bound established in the proof of Theorem~\ref{thm:wtf_lower},
\[
\alpha_u +\alpha_v \ge \min  \left\{  \int_0^1 g(x) dx, \min_{{p_u, x_v}} \{\int_0^{p_u}g(x)dx + (1-p_u)(1-g(x_v)) + \int_0^{x_v}g(x)dx \} \right\}.
\]
We are left to optimize function $g$ using the following linear program:
\begin{align*}
\max_g: \qquad &r\\
\text{s.t.}\qquad &r\leq  \int_0^x g(y) dy + \int_0^{z} g(y)dy + (1-x)(1-g(z)), \quad \forall x,z\in[0,1].
\end{align*}
After solving it, we remove redundant constraints with slacks and consider the following program:
\begin{align*}
(P) \quad \max_g: \qquad &r\\
\text{s.t.}\qquad & r\leq  \int_0^x g(y) dy + \int_0^{c-x} g(y)dy + (1-x)(1-g(c-x)), \quad \forall x\in[0,c],
\end{align*}
where $c=2-\sqrt{2}$. We know that the above two programs have the same optimal value. Moreover, as the program suggests, in order to construct a tight hard instance, all pairs $u,v$ matched in $\opt$ must satisfy $p_u + x_v = c$ when \wtf is run\footnote{Constraints must be tight almost everywhere. As otherwise, our primal dual analysis proves the competitive ratio of \wtf is strictly greater than $c$ on the specific instance.}. According to the instance structure and argument in Section~\ref{sec:wtfhardness}, it suffices to find a function $h:[0,1]\to[0,1]$ so that
\[
\int_0^x \frac{1-\tau(y)}{1-y+h(y)} = \tau(h(x)) = c - \tau(x), \quad \forall x\in[0,1].
\]

Here, $\tau:[0,1]\to[0,c]$ corresponds to the stationary water level and gives the first equation. Moreover, the perfect partner corresponding to $x$ also has water level $\tau(h(x))$ and we require it to be $c-\tau(x)$, which gives the second equation.
Therefore, $h(x)=\tau^{-1}(c-\tau(x))$. Let $f(x) = \tau^{-1}(x)$ and taking derivative over the above equation, it suffices to prove the existence of $f,\tau,h$ so that
\begin{align*}
\frac{1-\tau(x)}{1-x+h(x)}=-\tau'(x) & \Leftrightarrow \frac{1-\phi}{1-f(\phi)+f(c-\phi)}=-\frac{1}{f'(\phi)} \\ 
& \Leftrightarrow 1-f(\phi)+f(c-\phi)+(1-\phi)f'(\phi)=0, \quad \forall \phi\in[0,c].
\end{align*}

Now, consider the dual program of $P$ :
\begin{align*}
\min_q: \qquad & \int_0^c (1-x)q(x)dx\\
\text{s.t.}\qquad & 1-\int_0^c q(x) dx \leq 0, \\
& \int_0^x q(y)dy + \int_{c-x}^c q(y) dy - (1-x)q(x) \leq 0, \quad \forall x\in[0,c].
\end{align*}
According to primal dual theory, we know that the optimal dual solution $q(x)$ satisfies $\int_0^c q(x) dx = 1$ and $\int_0^x q(y)dy + \int_{c-x}^c q(y) dy + (1-x)q(x) = 0$.
Let $Q(x) = 1 - \int_0^x q(y)dy = \int_x^c q(y)dy$, we have
\begin{equation*}
1-Q(x)+Q(c-x)+(1-x)Q'(x)=0, \quad \forall x\in[0,c],
\end{equation*}
which is exactly the same equation we required for $f$.

\end{document}